\documentclass[aps,amsmath,a4paper,superscriptaddress,notitlepage,twoside,10pt,showpacs,showkeys]{revtex4-1}

\usepackage{geometry}
\usepackage{amssymb}
\usepackage{amsmath}
\usepackage{graphicx}
\usepackage{amsfonts}
\usepackage{hyperref}
\usepackage{color}
\usepackage{epsfig}
\usepackage{bbm}
\usepackage{color,enumerate,anysize,amsfonts,amsthm,mathrsfs,url}

\newcommand{\ket}[1]{\left| #1 \right>}

\newcommand{\cpx}[0]{\mathbb{C}}

\newcommand{\bvec}[1]{{\bf#1}}

\newcommand{\twomatrix}[1]{\left(\begin{array}{cc} #1 \end{array}\right)}

\newcommand{\1}{\openone}

\theoremstyle{plain} \newtheorem{thm}{Theorem}
\theoremstyle{plain} 
\theoremstyle{definition} \newtheorem*{rmk}{Remark}
\theoremstyle{definition} 
\theoremstyle{plain} \newtheorem{lemma}[thm]{Lemma}
\theoremstyle{definition} 
\theoremstyle{plain} 
\theoremstyle{plain}

\begin{document}

\title{Achievability of two qubit gates using linear optical elements and post-selection}

\author{Josh Cadney}
\affiliation{School of Mathematics, University of Bristol, Bristol BS8 1TW, United Kingdom}

\date{31 March 2014}

\begin{abstract}
We study the class of two qubit gates which can be achieved using only linear optical elements (beam splitters and phase shifters) and post-selection. We are able to exactly characterize this set, and find that it is impossible to implement most two qubit gates in this way. The proof also gives rise to an algorithm for calculating the optimal success probability of those gates which are achievable.
\end{abstract}

\maketitle

\section{Introduction} \label{intro}
Linear optical quantum computing is a promising architecture for building a universal quantum computer, due to the high fidelity of linear optical elements (beam splitters and phase shifters) and the insensitivity of photons to decoherence. It has been shown \cite{KLM01} that optical systems are indeed universal for quantum computation if it is possible to implement (near perfectly): linear optical elements, single photon sources, photon number detectors and adaptive feedback. Unfortunately, the implementation of all these things simultaneously is still far off. In particular, the use of feedback within an optical circuit, and producing photon sources with a high probability of success, are very challenging.

Consequently, many cutting-edge experiments in the field focus on demonstrating some subset of these resources, as a proof of principle. The resulting quantum circuits are often allowed to succeed, with probability $p<1$, conditioned on certain measurement outcomes. In this paper, we focus on a method of performing quantum gates which uses only linear optical elements and post-selection (outlined in Figure \ref{circuit_1}). This method has been used to demonstrate a CNOT gate which succeeds with probability $\frac19$ \cite{RLB02,HT02,OPW03}, a reconfigurable controlled two-qubit operation \cite{SVP12,LPN11} and in small-scale versions of Shor's algorithm \cite{PMO09,MLL12}.

We study these experiments from a theoretical point of view, focusing on two questions: `what gates can we perform using this method?' and `what is the optimal probability of success?'. The second question has been studied before; in \cite{KOE10} the optimal success probability for controlled-phase gates was derived, and a framework for solving the problem for a general gate was discussed in \cite{Kie08}. However, it seems that the first, more fundamental, question has so far been ignored. This project focuses on two qubit gates, and it was our initial aim to design a circuit which could perform an arbitrary two qubit gate. However, we found that this is not possible. In fact, our results show that \emph{almost all} two qubit gates cannot be achieved in this scenario.

We proceed as follows: in section \ref{problem} we detail the set up we are studying, and outline the problem. In section \ref{result} we present our main result: a complete characterization of those two qubit gates which can be achieved in this set up. In section \ref{success_prob} we give an algorithm for computing the optimal success probability of those gates which can be performed. Finally in section \ref{open} we mention some open problems.

\section{The Problem} \label{problem}

\begin{figure}[t]
	\includegraphics[width=0.5\textwidth]{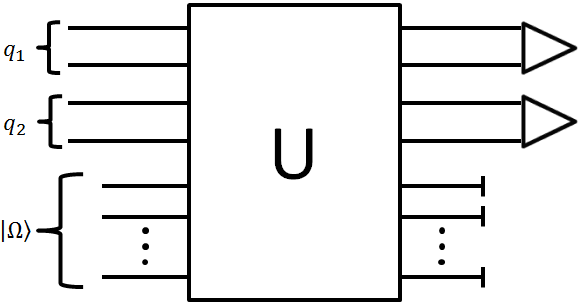}
	\caption{The circuit we consider, containing $2$ photons in $N$ modes. Initially, the first pair of modes contains one photon, encoding one logical qubit. Similarly, the second pair of modes encodes a second logical qubit. The remaining $N-4$ auxiliary modes are initially empty. The circuit $U$ is an arbitrary series of beam splitters and phase shifters. After performing $U$ we post-select on finding one photon in each of the first two pairs of modes. We then disregard the auxiliary modes.}
	\label{circuit_1}
\end{figure}

The circuit we consider is shown in Figure \ref{circuit_1}. We have two photons in $N$ modes. Initially, the first two modes contain one photon, which encodes a logical qubit via a dual rail encoding. Similarly, modes $3$ and $4$ contain the second photon, which encodes our second logical qubit. The remaining $N-4$ auxiliary modes are empty. Concretely, let $\ket{\Omega}$ denote the vacuum state, and $a_0^\dagger, \ldots, a_{N-1}^\dagger$ denote the creation operators of the $N$ modes. Then the four computational basis states, which correspond to the logical states $\ket{00}, \ket{01}, \ket{10}, \ket{11}$ respectively, are
\begin{equation}
	a_0^\dagger a_2^\dagger \ket\Omega, a_0^\dagger a_3^\dagger \ket\Omega, a_1^\dagger a_2^\dagger \ket\Omega, a_1^\dagger a_3^\dagger \ket\Omega
\end{equation}
We refer to the span of these states as the \emph{computational subspace}, and we assume that the initial state of the circuit is in the computational subspace.

The initial state is then acted on by the linear optical component $U$, which is allowed to be any sequence of beam splitters and phase shifters. It is well known that any unitary transformation of modes can be achieved in this way \cite{RZBB94}. Therefore, the effect of this component is to map
\begin{equation}
	a_i^\dagger \to \sum_{j=0}^{N-1} u_{ji} a_j^\dagger
\end{equation}
for some $N \times N$ unitary matrix (also denoted by $U$) with entries $u_{ij}$. This means that the effect of the component $U$ on a computational basis state $a_i^\dagger a_j^\dagger \ket{\Omega}$ is as follows
\begin{equation} \label{U_computational}
	a_i^\dagger a_j^\dagger \ket{\Omega} \to \sum_{k,l=0}^{N-1} u_{ki} u_{lj} a_k^\dagger a_l^\dagger \ket{\Omega}
\end{equation}

In the final stage of the circuit, we discard the $N-4$ auxiliary modes, and we post-select on finding the state in the computational subspace. This is equivalent to requiring that there is exactly one photon in each of the first two pairs of modes. This requires us to perform a measurement on each pair of modes. In theory this measurement can be performed non-destructively \cite{KLD02} meaning that the resulting state can then be passed on to future operations, however this requires at least two additional photons. In current experiments the state is usually destroyed at this time. Mathematically, we model this step as a projection onto the computational subspace.

Putting this together, by applying the post-selection to the resulting state in \eqref{U_computational}, we find that the effect of the circuit on a computational basis state is given by
\begin{equation} \label{circuit_output}
\begin{split}
    a_i^\dagger a_j^\dagger \ket{\Omega} \to \left[ (u_{0i}u_{2j} + u_{2i}u_{0j}) a_0^\dagger a_2^\dagger \right. + (u_{0i}u_{3j} &+ u_{3i}u_{0j}) a_0^\dagger a_3^\dagger \\
    + (u_{1i}u_{2j} + u_{2i}u_{1j}) a_1^\dagger a_2^\dagger &+ \left. (u_{1i}u_{3j} + u_{3i}u_{1j}) a_1^\dagger a_3^\dagger \right] \ket{\Omega},
\end{split}
\end{equation}
so, for example, in the notation of the computational subspace we have
\begin{equation}
	\ket{00} \to (u_{00}u_{22} + u_{20}u_{02})\ket{00} + (u_{00}u_{32} + u_{30}u_{02})\ket{01} + (u_{10}u_{22} + u_{20}u_{12})\ket{10} + (u_{10}u_{32} + u_{30}u_{12})\ket{11}.
\end{equation}

We now consider how to implement a two qubit gate in the computational subspace, using the circuit in Figure \ref{circuit_1}. Notice that the resulting states in \eqref{circuit_output} only depend on the values $u_{ij}$ with $0\leq i,j \leq 3$. With this in mind we define the matrix $\tilde{U}$ to be the upper left corner of the matrix $U$:
\begin{equation}
	\tilde{U} := \left( \begin{array}{cccc}
		u_{00} & u_{01} & u_{02} & u_{03} \\
		u_{10} & u_{11} & u_{12} & u_{13} \\
		u_{20} & u_{21} & u_{22} & u_{23} \\
		u_{30} & u_{31} & u_{32} & u_{33}
	\end{array} \right)
\end{equation}
and we define a matrix-valued function, $f$, such that
\begin{equation} \label{function_def}
	f(\tilde{U}) := \left( \begin{array}{cccc}
		u_{00}u_{22} + u_{20}u_{02} & u_{00}u_{23} + u_{20}u_{03} & u_{01}u_{22} + u_{21}u_{02} & u_{01}u_{23} + u_{21}u_{03} \\
		u_{00}u_{32} + u_{30}u_{02} & u_{00}u_{33} + u_{30}u_{03} & u_{01}u_{32} + u_{31}u_{02} & u_{01}u_{33} + u_{31}u_{03} \\
		u_{10}u_{22} + u_{20}u_{12} & u_{10}u_{23} + u_{20}u_{13} & u_{11}u_{22} + u_{21}u_{12} & u_{11}u_{23} + u_{21}u_{13} \\
		u_{10}u_{32} + u_{30}u_{12} & u_{10}u_{33} + u_{30}u_{13} & u_{11}u_{32} + u_{31}u_{12} & u_{11}u_{33} + u_{31}u_{13}
	\end{array} \right)
\end{equation}
The idea is that $f(\tilde{U})$ is the transformation induced on the computational subspace by the circuit in Figure \ref{circuit_1}. More precisely, suppose that we wish to implement the unitary matrix, $W$, in the computational subspace, with a probability of success, $p$. Let $W$ take the form
\begin{equation}
	W= \left( \begin{array}{cccc}
		w_{00} & w_{01} & w_{02} & w_{03} \\
		w_{10} & w_{11} & w_{12} & w_{13} \\
		w_{20} & w_{21} & w_{22} & w_{23} \\
		w_{30} & w_{31} & w_{32} & w_{33}
	\end{array} \right)
\end{equation}
Then we need to find $\tilde{U}$ such that we have
\begin{equation} \label{matrix_equation}
	\sqrt{p} W = f(\tilde{U})
\end{equation}
subject to the constraint that the matrix $\tilde{U}$ forms the upper left corner of a unitary matrix. Notice that if we have $\tilde{U}$ such that $f(\tilde{U})=\sqrt{p}W$ then $f(p^{-\frac14} \tilde{U}) = W$, and so all solutions of \eqref{matrix_equation} are a constant multiple of solutions of the equation
\begin{equation} \label{simple_equation}
	f(\tilde{U}) = W.
\end{equation}
Furthermore, it is known \cite{Kie08} that the matrix $\tilde{U}$ can be written as the upper left corner of a unitary matrix if and only if its singular values are at most $1$. Write $s_1(M)$ for the largest singular value of a matrix $M$. Suppose we are given an arbitrary matrix $\tilde{U}$ which is a solution to \eqref{simple_equation}. Then, either $s_1(\tilde{U})\leq 1$ and $\tilde{U}$ is a solution to \eqref{matrix_equation} with $p=1$, or the matrix $s_1(\tilde{U})^{-1} \tilde{U}$ is a solution to \eqref{matrix_equation} with $p=s_1(\tilde{U})^{-\frac14}$. Consequently, when we are only interested in the existence of solutions to \eqref{matrix_equation} \emph{for any value of $p$}, we need only consider the existence of solutions to \eqref{simple_equation}.

\subsection*{Invariance under local unitaries}

In this section we note that if \eqref{simple_equation} has a solution for a given $W$, then it also has a solution for any matrix of the form $W':=(V_1\otimes V_2)W(V_3\otimes V_4)$ where $V_1, V_2, V_3, V_4$ are $2\times 2$ unitary matrices. (We say that a matrix $W'$ of this form is \emph{locally equivalent} to $W$).

The reason for this is as follows. Let $X$ and $Y$ be the block matrices
\begin{equation}
\begin{aligned}
	X:= \left( \begin{array}{cc} V_1 & 0 \\ 0 & V_2 \end{array} \right) \\
	Y:= \left( \begin{array}{cc} V_3 & 0 \\ 0 & V_4 \end{array} \right)
\end{aligned}
\end{equation}
Then the following relation holds:
\begin{equation} \label{invariance_equation}
	f(X\tilde{U} Y) = (V_1 \otimes V_2) f(\tilde{U}) (V_3 \otimes V_4)
\end{equation}
From this it is clear that if there exists $\tilde{U}$ such that $f(\tilde{U})=W$ then there exists also $\tilde{U}' = X\tilde{U} Y$ such that $f(\tilde{U}') = W'$.

\begin{figure}[t]
	\includegraphics[width=0.5\textwidth]{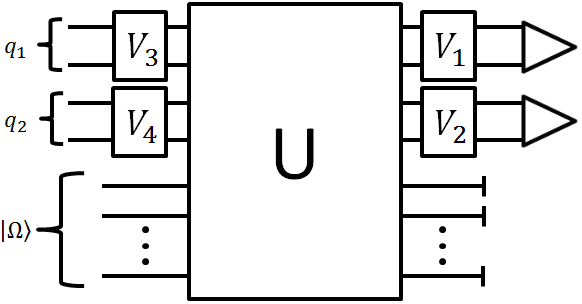}
	\caption{The modified circuit. This circuit is identical to Figure \ref{circuit_1}, except for the addition of the components $V_1,\ldots,V_4$. Each of these components is a series of beam splitters and phase shifters, and acts exclusively on one of the logical qubits. The effect of these components is to perform local unitaries on the computational subspace, before and after performing the circuit $U$.}
	\label{circuit_2}
\end{figure}

Equation \eqref{invariance_equation} can be verified in two ways. First, we will give a physically motivated argument. Consider a circuit of the form shown in Figure \ref{circuit_2}, where the optical component $U$ is such that it implements a unitary $W$ on the computational subspace with probability $p$ (i.e. $f(\tilde{U})=\sqrt{p}W$). Now suppose that this circuit is applied to a state $\ket{\psi}$ in the computational subspace. The first local unitaries will map the state to $(V_3\otimes V_4)\ket{\psi}$. Then the component $U$ will map it to $\sqrt{p}W(V_3\otimes V_4)\ket{\psi} + \sqrt{1-p} \ket{\alpha^\perp}$ where $\ket{\alpha^\perp}$ is a state orthogonal to the computational subspace. The final local unitaries will map this to $\sqrt{p}(V_1\otimes V_2) W (V_3\otimes V_4) \ket{\psi} + \sqrt{1-p}\ket{\beta^\perp}$. Overall, the circuit performs the unitary $XUY$ (on modes). Therefore, we conclude that $f(X\tilde{U}Y)= \sqrt{p} (V_1\otimes V_2) W (V_3\otimes V_4)$.

For a more mathematical proof, notice that if we write $\tilde{U}$ as a block matrix:
\begin{equation}
	\tilde{U} = \left( \begin{array}{cc} A & B \\ C & D \end{array} \right)
\end{equation}
then we have
\begin{equation}
	f(\tilde{U}) = A\otimes D + (B\otimes C) S
\end{equation}
where $S$ is the swap operator given by
\begin{equation}
	S:= \left(\begin{array}{cccc} 1&0&0&0 \\ 0&0&1&0 \\ 0&1&0&0 \\ 0&0&0&1 \end{array}\right)
\end{equation}
Consequently,
\begin{equation}
\begin{aligned}
	f(X\tilde{U}Y) &= f\left( \left(\begin{array}{cc} V_1AV_3 & V_1BV_4 \\ V_2CV_3 & V_2DV_4 \end{array}\right) \right) \\
	&= (V_1\otimes V_2)(A\otimes D)(V_3\otimes V_4) + (V_1\otimes V_2)(B\otimes C)(V_4\otimes V_3) S \\
	&= (V_1\otimes V_2)(A\otimes D)(V_3\otimes V_4) + (V_1\otimes V_2)(B\otimes C) S(V_3\otimes V_4) \\
	&= (V_1 \otimes V_2) f(\tilde{U}) (V_3 \otimes V_4)
\end{aligned}
\end{equation}
For the third equality here we used the relation $(Q\otimes P)S = S(P\otimes Q)$ which holds for all $2\times 2$ matrices $P$ and $Q$.

\section{The main result} \label{result}

In this section we will present our main result. Let us begin from a well known decomposition of two qubit gates.

\begin{lemma}[{\cite[Appendix A]{KC01}}] \label{cartan_decomp}
	Let $W$ be a $4\times 4$ unitary matrix. Then there exist $\alpha, \beta, \gamma \in [0,2\pi]$ and $2 \times 2$ unitary matrices $V_1, V_2, V_3, V_4$ such that
	\begin{equation}
		W = (V_1\otimes V_2) \exp \left[ i\alpha X \otimes X + i\beta Y\otimes Y + i\gamma Z\otimes Z \right] (V_3\otimes V_4)
	\end{equation}
	where $X,Y$ and $Z$ are the Pauli matrices:
	\begin{equation}
		X=\twomatrix{0&1\\1&0}, Y=\twomatrix{0&-i\\i&0}, Z=\twomatrix{1&0\\0&-1}.
	\end{equation}
\end{lemma}
\begin{rmk}
	This decomposition is not unique. The matrix can be written in this form for multiple triples $(\alpha,\beta,\gamma)$.
\end{rmk}

In the previous section we showed that the achievability of a gate under our scheme is invariant under local unitaries. Consequently, Lemma \ref{cartan_decomp} tells us that we need only consider unitaries of the form $\exp \left[ i\alpha X \otimes X + i\beta Y\otimes Y + i\gamma Z\otimes Z \right]$ in order to develop a complete picture. Written in matrix form we have
\begin{equation} \label{cartan_form}
\begin{aligned}
	e^{i\alpha X \otimes X + i\beta Y\otimes Y + i\gamma Z\otimes Z} &= 
	\left(\begin{array}{cccc} 
		e^{i\gamma}\cos(\alpha-\beta) & 0 & 0 & ie^{i\gamma}\sin(\alpha-\beta) \\ 
		0 & e^{-i\gamma}\cos(\alpha+\beta) & ie^{-i\gamma}\sin(\alpha+\beta) & 0 \\ 
		0 & ie^{-i\gamma}\sin(\alpha+\beta) & e^{-i\gamma}\cos(\alpha+\beta) & 0 \\ 
		ie^{i\gamma}\sin(\alpha-\beta) & 0 & 0 & e^{i\gamma}\cos(\alpha-\beta) 
	\end{array}\right) \\
	&=\left( \begin{array}{cccc} w_1&0&0&w_4 \\ 0&w_2&w_3&0 \\ 0&w_3&w_2&0 \\ w_4&0&0&w_1 \end{array} \right)
\end{aligned}
\end{equation}
where $w_1=e^{i\gamma}\cos(\alpha-\beta)$, $w_2=e^{-i\gamma}\cos(\alpha+\beta)$, $w_3=ie^{-i\gamma}\sin(\alpha+\beta)$ and $w_4=ie^{i\gamma}\sin(\alpha-\beta)$. We are interested in solving equation \eqref{simple_equation} for this choice of $W$.

\begin{thm} \label{main_theorem}
	Let $W$ be a $4\times 4$ unitary matrix, which is locally equivalent to a matrix of the form \eqref{cartan_form}. Then $W$ can be achieved by the scheme in Figure \ref{circuit_1} if and only if at least one of the six values $\alpha\pm\beta, \alpha\pm\gamma, \beta\pm\gamma$ is equal to $0$ or $\frac{\pi}{2}$ modulo $\pi$.
\end{thm}
\begin{proof}
	According to lemmas \ref{non-zero_case} and \ref{zero_case} (see Appendix) $W$ can be achieved if and only if either
	\begin{equation} \label{signs_equation}
		w_1 \pm w_2 \pm w_3 \pm w_4 = 0
	\end{equation}
	for at least one of the eight possible choices for the signs, or
	\begin{equation}
		w_i = 0
	\end{equation}
	for some $i$.
	
	Let us consider the second case first. In order to have $w_i=0$ for some $i$, we must have one of $\cos(\alpha\pm\beta), \sin(\alpha\pm\beta)$ equal to zero. In other words, we must have one of $\alpha\pm\beta$ equal to $0$ or $\frac{\pi}{2}$ modulo $\pi$.
	
	Now consider the first case. Suppose that we have
	\begin{equation}
		w_1 + w_2 + w_3 + w_4 = 0
	\end{equation}
	This means that
	\begin{equation} \label{ppp_constraint}
		e^{i\gamma}\cos(\alpha-\beta) + e^{-i\gamma}\cos(\alpha+\beta) + ie^{-i\gamma}\sin(\alpha+\beta) + ie^{i\gamma}\sin(\alpha-\beta) = 0
	\end{equation}
	Let us write $c_\theta, s_\theta$ as a shorthand for $\cos\theta, \sin\theta$ in order to simplify notation. Then we can expand \eqref{ppp_constraint} to give
	\begin{equation}
		(c_\gamma+is_\gamma)(c_\alpha c_\beta + s_\alpha s_\beta) + (c_\gamma-is_\gamma)(c_\alpha c_\beta - s_\alpha s_\beta) + (ic_\gamma + s_\gamma)(s_\alpha c_\beta + c_\alpha s_\beta) + (ic_\gamma - s_\gamma)(s_\alpha c_\beta - c_\alpha s_\beta) = 0 
	\end{equation}
	Splitting this into real and imaginary parts we have
	\begin{equation}
	\begin{aligned}
		c_\gamma (c_\alpha c_\beta + s_\alpha s_\beta + c_\alpha c_\beta - s_\alpha s_\beta) + s_\gamma (s_\alpha c_\beta + c_\alpha s_\beta - s_\alpha c_\beta + c_\alpha s_\beta) &= 0 \\
		s_\gamma (c_\alpha c_\beta + s_\alpha s_\beta - c_\alpha c_\beta + s_\alpha s_\beta) + c_\gamma (s_\alpha c_\beta + c_\alpha s_\beta + s_\alpha c_\beta - c_\alpha s_\beta) &= 0
	\end{aligned}
	\end{equation}
	which implies
	\begin{equation}
	\begin{aligned}
		2c_\gamma c_\alpha c_\beta + 2s_\gamma c_\alpha s_\beta &= 0 \\
		2s_\gamma s_\alpha s_\beta + 2c_\gamma s_\alpha c_\beta &= 0
	\end{aligned}
	\end{equation}
	and hence
	\begin{equation}
	\begin{aligned}
		c_\alpha \cos(\beta-\gamma) &= 0 \\
		s_\alpha \cos(\beta-\gamma) &= 0
	\end{aligned}
	\end{equation}
	Both of these equations can be satisfied only when $\cos(\beta-\gamma)=0$, or equivalently, when $\beta-\gamma$ is equal to $\frac{\pi}{2}$ modulo $\pi$. In exactly the same way, if we expand the other seven choices for the signs in \eqref{signs_equation} then we obtain the conditions
	\begin{equation}
	\begin{aligned}
		\cos(\beta+\gamma) &= 0 \\
		\sin(\beta\pm\gamma) &= 0 \\
		\cos(\alpha\pm\gamma) &= 0 \\
		\sin(\alpha\pm\gamma) &= 0
	\end{aligned}
	\end{equation}
\end{proof}

\section{Optimal success probability} \label{success_prob}
The result of the previous section shows that the circuit in Figure \ref{circuit_1} cannot implement almost all two qubit gates, for any probability of success. However, it also implies that the set of gates which can be achieved has 15 independent real-valued parameters (as opposed to 16 for the $4\times 4$ unitary group). Therefore, there are many gates which can be implemented by this scheme, and, in fact, this set contains many important gates, including CNOT and all controlled phase gates.

This means that this set up still has value in an experimental setting, and raises another important question: for those gates which can be achieved, what is the maximum probability of success with which they succeed? Looking carefully at the proof of Lemma \ref{non-zero_case} we see that we actually found all solutions of equation \eqref{simple_equation} for those cases where a solution exists. The family of solutions is characterized by two free complex-valued parameters (and one free parameter which takes values $\pm1$).

This leads us to an algorithm for computing the optimal probability of success. Given a unitary $W$ which we wish to implement, convert it into the form \eqref{cartan_form} by applying local unitaries. There are many well known algorithms (for example \cite{BdV08}) which can accomplish this.

Now, suppose that we have $\tilde{U}$ which is a solution of \eqref{simple_equation}. Then, according to the comment below \eqref{simple_equation}, this gives us an implementation of the gate $W$ with success probability $s_1(\tilde{U})^{-\frac14}$. It is simple to write a function (call it $g$) which calculates this probability. Running a numerical optimization of $g$ over the entire family of solutions will result in finding the optimal success probability. Since we have an explicit characterization of the family of solutions, many standard numerical routines are suitable for this purpose. For example, we used the BFGS method (see \cite{NW99}), which is implemented in the optimize package of the SciPy library \cite{SciPy}. For the problem at hand, this method converges within seconds on a standard desktop computer, although it does not guarantee finding the best solution.


\section{Open questions} \label{open}
We have considered the problem of implementing two qubit gates under a contemporary scheme for experiments in linear optical quantum computation. Our results show that most such gates cannot be performed within this scheme, with any probability of success. This begs the question: why does this scheme support some gates, but not others? Is there a physical consideration which sets these gates apart? Is there some physical meaning to the necessary and sufficient condition given in Theorem \ref{main_theorem}? We leave this question open. Another obvious extension of this work would be to consider the case of three qubit gates and higher. Note that in this scenario our approach becomes very complicated. Indeed, we would then need to solve a system of 64 cubic equations in 64 unknowns.

\section*{Acknowledgments}
The author wishes to give special thanks to Enrique Mart\'in Lopez for initially drawing attention to the problem, insightful discussions throughout the project, and for checking the manuscript. I also acknowledge useful conversations with Noah Linden and (indirectly) with Anthony Laing. This work was supported by the U.K. EPSRC.

\newpage
\appendix \label{appendix}
\section{}
\begin{lemma} \label{non-zero_case}
	A unitary $W$ of the form \eqref{cartan_form} with $w_1,w_2,w_3,w_4 \neq0$ can be achieved by the scheme in Figure \ref{circuit_1} if and only if
	\begin{equation}
		w_1 \pm w_2 \pm w_3 \pm w_4 = 0
	\end{equation}
	for at least one of the eight possible choices for the signs ($\pm$).
\end{lemma}

\begin{proof}
	We are interested in finding solutions of the equation \eqref{simple_equation}. More explicitly, we are looking for solutions to
	\begin{equation}
		\left( \begin{array}{cccc}
		u_{00}u_{22} + u_{20}u_{02} & u_{00}u_{23} + u_{20}u_{03} & u_{01}u_{22} + u_{21}u_{02} & u_{01}u_{23} + u_{21}u_{03} \\
		u_{00}u_{32} + u_{30}u_{02} & u_{00}u_{33} + u_{30}u_{03} & u_{01}u_{32} + u_{31}u_{02} & u_{01}u_{33} + u_{31}u_{03} \\
		u_{10}u_{22} + u_{20}u_{12} & u_{10}u_{23} + u_{20}u_{13} & u_{11}u_{22} + u_{21}u_{12} & u_{11}u_{23} + u_{21}u_{13} \\
		u_{10}u_{32} + u_{30}u_{12} & u_{10}u_{33} + u_{30}u_{13} & u_{11}u_{32} + u_{31}u_{12} & u_{11}u_{33} + u_{31}u_{13}
		\end{array} \right)
		= \left( \begin{array}{cccc} w_1&0&0&w_4 \\ 0&w_2&w_3&0 \\ 0&w_3&w_2&0 \\ w_4&0&0&w_1 \end{array} \right)
	\end{equation}
	which is a system of 16 polynomial equations in the variables $u_{00},\ldots,u_{33}$.
	
	We will first consider the case in which $w_i\neq0$ for each $i$. This corresponds to the case in which $\alpha\pm\beta$ is not $0$ or $\frac{\pi}{2}$ modulo $\pi$.
	The first key observation is that in any solution, none of the $u_{ij}$ can be zero. For example, suppose that $u_{00}=0$. There are four equations containing $u_{00}$:
	\begin{align}
		u_{00}u_{22} + u_{20}u_{02} &= w_1 \label{u00one}\\
		u_{00}u_{23} + u_{20}u_{03} &= 0 \label{u00two}\\
		u_{00}u_{32} + u_{30}u_{02} &= 0 \label{u00three}\\
		u_{00}u_{33} + u_{30}u_{03} &= w_2 \label{u00four}
	\end{align}
	If $u_{00}=0$ then \eqref{u00two} implies that either $u_{20}$ or $u_{03}$ must also be zero. But if $u_{20}=0$ then \eqref{u00one} is false, and if $u_{03}=0$ then \eqref{u00four} is false. Consequently, any solution of this system of equations must have $u_{00}\neq0$. An identical argument shows that in fact we must have $u_{ij}\neq 0$ for all $0\leq i,j\leq 3$.

	The second key observation we make is that the 8 equations with 0 on the right-hand-side can be written as follows:
	\begin{align}
		\left(\begin{array}{cccc} u_{32}&0&u_{30}&0 \\ u_{23}&0&0&u_{20} \\ 0&u_{22}&u_{21}&0 \\ 0&u_{33}&0&u_{31} \end{array}\right) \left(\begin{array}{c} u_{00}\\u_{01}\\u_{02}\\u_{03} \end{array}\right) &= \left(\begin{array}{c} 0\\0\\0\\0 \end{array}\right) \\
		\left(\begin{array}{cccc} u_{22}&0&u_{20}&0 \\ u_{33}&0&0&u_{30} \\ 0&u_{32}&u_{31}&0 \\ 0&u_{23}&0&u_{21} \end{array}\right) \left(\begin{array}{c} u_{10}\\u_{11}\\u_{12}\\u_{13} \end{array}\right) &= \left(\begin{array}{c} 0\\0\\0\\0 \end{array}\right)
	\end{align}
	Let us write these equations as
	\begin{align}
		M_1 \bvec{u_1} &= \bvec{0} \\
		M_2 \bvec{u_2} &= \bvec{0}
	\end{align}
	Now, in order for these equations to have non-zero solutions for $\bvec{u_1}$ and $\bvec{u_2}$ we require $M_1$ and $M_2$ to be singular matrices. Thus we must have
	\begin{equation} \label{det_constraint}
		\det(M_1)=u_{22}u_{23}u_{30}u_{31} - u_{20}u_{21}u_{32}u_{33} = 0.
	\end{equation}
	The condition $\det(M_2)=0$ yields the same constraint.
	
	Furthermore, we can also conclude that the vector $\bvec{u_1}$ is a non-zero element of the kernel of $M_1$. Making use of the constraint \eqref{det_constraint} and applying Gaussian elimination, we find that the reduced row echelon form of $M_1$ is
	\begin{equation}
		\left(\begin{array}{cccc}
			1 & 0 & 0 & \frac{u_{20}}{u_{23}} \\
			0 & 1 & 0 & \frac{u_{31}}{u_{33}} \\
			0 & 0 & 1 & -\frac{u_{20}u_{32}}{u_{23}u_{30}} \\
			0 & 0 & 0 & 0
		\end{array}\right)
	\end{equation}
	From this we can conclude that for some non-zero $\lambda\in\cpx$ we have
	\begin{equation}
		\bvec{u_1} = \left(\begin{array}{c} u_{00}\\u_{01}\\u_{02}\\u_{03} \end{array}\right)
		= \lambda \left(\begin{array}{c} -\frac{u_{20}}{u_{23}} \\ -\frac{u_{31}}{u_{33}} \\ \frac{u_{20}u_{32}}{u_{23}u_{30}} \\ 1 \end{array}\right)
	\end{equation}
	By an identical argument applied to $M_2$, we can conclude also that for some non-zero $\mu\in\cpx$ we have
	\begin{equation}
		\bvec{u_2} = \left(\begin{array}{c} u_{10}\\u_{11}\\u_{12}\\u_{13} \end{array}\right)
		= \mu \left(\begin{array}{c} -\frac{u_{30}}{u_{33}} \\ -\frac{u_{21}}{u_{23}} \\ \frac{u_{22}u_{30}}{u_{20}u_{33}} \\ 1 \end{array}\right)
	\end{equation}
	
	We have now reduced our original system of 16 equations in 16 variables to a system of 9 equations in 10 variables ($\lambda,\mu,u_{20},\ldots,u_{33}$). Eight of the remaining equations are those corresponding to the non-zero matrix elements of $W$, and they too can be expressed in terms of $M_1$ and $M_2$:
	\begin{align}
		M_2 \bvec{u_1} &= \left(\begin{array}{c} w_1\\w_2\\w_3\\w_4 \end{array}\right) \\
		M_1 \bvec{u_2} &= \left(\begin{array}{c} w_4\\w_3\\w_2\\w_1 \end{array}\right)
	\end{align}
	(The other equation is the constraint \eqref{det_constraint}).
	Expanding, and substituting the expressions obtained for $\bvec{u_1}$ and $\bvec{u_2}$ above we get
	\begin{align}
		\lambda \left(\begin{array}{c} -\frac{u_{20}u_{22}}{u_{23}} + \frac{u_{20}^2u_{32}}{u_{23}u_{30}} \\ -\frac{u_{20}u_{33}}{u_{23}} + u_{30} \\ -\frac{u_{31}u_{32}}{u_{33}} + \frac{u_{20}u_{31}u_{32}}{u_{23}u_{30}} \\ -\frac{u_{23}u_{31}}{u_{33}} + u_{21} \end{array}\right) &= \left(\begin{array}{c} w_1\\w_2\\w_3\\w_4 \end{array}\right) \\
		\mu \left(\begin{array}{c} -\frac{u_{30}u_{32}}{u_{33}} + \frac{u_{22}u_{30}^2}{u_{20}u_{33}} \\ -\frac{u_{23}u_{30}}{u_{33}} + u_{20} \\ -\frac{u_{21}u_{22}}{u_{23}} + \frac{u_{21}u_{22}u_{30}}{u_{20}u_{33}} \\ -\frac{u_{21}u_{33}}{u_{23}} + u_{31} \end{array}\right) &= \left(\begin{array}{c} w_4\\w_3\\w_2\\w_1 \end{array}\right)
	\end{align}
	Here we have obtained 2 distinct expressions for each of $w_1,w_2,w_3$ and $w_4$. Equating the two expressions for $w_3$ gives
	\begin{equation}
		\lambda \left( -\frac{u_{31}u_{32}}{u_{33}} + \frac{u_{20}u_{31}u_{32}}{u_{23}u_{30}} \right) = \mu \left( -\frac{u_{23}u_{30}}{u_{33}} + u_{20} \right)
	\end{equation}
	This implies
	\begin{equation}
		\lambda \frac{u_{31}u_{32}}{u_{23}u_{30}u_{33}} \left( -u_{23}u_{30} + u_{20}u_{33} \right)
		= \mu \frac{1}{u_{33}} \left( -u_{23}u_{30} + u_{20}u_{33} \right)
	\end{equation}
	and thus
	\begin{equation} \label{w3_constraint}
		\frac{\lambda}{\mu} = \frac{u_{23}u_{30}}{u_{31}u_{32}}
	\end{equation}
	Similarly, equating the 2 expressions for $w_4$, and making use of \eqref{det_constraint}, gives
	\begin{equation}
		\frac{\lambda}{\mu} = \frac{u_{30}u_{32}}{u_{23}u_{31}}
	\end{equation}
	which, combined with \eqref{w3_constraint} gives
	\begin{equation} \label{w4_constraint}
		u_{23}^2 = u_{32}^2
	\end{equation}
	In the same way, equating the expressions for $w_1$ and $w_2$ leads to
	\begin{align}
		u_{20}^2 &= u_{31}^2 \\
		u_{22}^2 &= u_{33}^2
	\end{align}
	
	We could now summarize our progress, by restating the problem in the following way
	\begin{equation}
	\begin{aligned}
		&\text{find} & & u_{20},u_{21},u_{22},u_{23},u_{30},u_{31},u_{32},u_{33},\lambda,\mu \\
		&\text{subject to}
		& & \mu\left(-\frac{u_{21}u_{33}}{u_{23}} + u_{31}\right) = w_1 \\
		&&& \lambda\left(-\frac{u_{20}u_{33}}{u_{23}} + u_{30}\right)= w_2 \\
		&&& \mu\left(-\frac{u_{23}u_{30}}{u_{33}} + u_{20}\right) = w_3 \\
		&&& \lambda\left(-\frac{u_{23}u_{31}}{u_{33}} + u_{21}\right)= w_4 \\
		&&& u_{23}^2 = u_{32}^2 \\
		&&& u_{22}^2 = u_{33}^2 \\
		&&& u_{20}^2 = u_{31}^2 \\
		&&& \frac{\lambda}{\mu} = \frac{u_{30}u_{32}}{u_{23}u_{31}} \\
		&&& u_{22}u_{23}u_{30}u_{31} = u_{20}u_{21}u_{32}u_{33}
	\end{aligned}
	\end{equation}
	We will attack this problem via a series of substitutions. First, we introduce a new variable $\alpha$ and eliminate $u_{33}$ by setting
	\begin{equation}
		u_{33} = \alpha u_{23}
	\end{equation}
	This will simplify the notation somewhat. Next we rearrange the first and second constraints to eliminate $u_{21}$ and $u_{20}$
	\begin{align}
		u_{21} &= \alpha^{-1} (u_{31}-\mu^{-1}w_1) \label{u21_def} \\
		u_{20} &= \alpha^{-1} (u_{30}-\lambda^{-1}w_2) \label{u20_def}
	\end{align}
	Substituting these expressions into the third and fourth constraints gives us
	\begin{align}
		-\mu w_2 &= \alpha\lambda w_3 \\
		-\lambda w_1 &= \alpha\mu w_4
	\end{align}
	Now we use the fifth constraint to eliminate $u_{32}$
	\begin{equation}
		u_{32} = b_1 u_{23}
	\end{equation}
	where we have introduced a new variable $b_1$ which can only take the values $\pm 1$. Similarly, we can use the sixth constraint to eliminate $u_{22}$:
	\begin{equation}
		u_{22} = b_2 \alpha u_{23}
	\end{equation}
	Finally, substituting the above into the final three constraints gives
	\begin{align}
		\alpha^2 u_{31}^2 &= (u_{30}-\lambda^{-1}w_2)^2 \\
		\frac{\lambda}{\mu} &= b_1 \frac{u_{30}}{u_{31}} \\
		b_1 b_2 \alpha^2 u_{30}u_{31} &= (u_{30}-\lambda^{-1}w_2)(u_{31}-\mu^{-1}w_1) \label{u20u21notzero}
	\end{align}
	
	We have now reduced the problem to the following
	\begin{equation}
	\begin{aligned}
		&\text{find} & & u_{23},u_{30},u_{31},\alpha,\lambda,\mu\in\cpx, b_1,b_2\in\{-1,1\} \\
		&\text{subject to}
		& & -\mu w_2 = \alpha\lambda w_3 \\
		&&& -\lambda w_1 = \alpha\mu w_4 \\
		&&& \alpha^2 u_{31}^2 = (u_{30}-\lambda^{-1}w_2)^2 \\
		&&& \frac{\lambda}{\mu} = b_1 \frac{u_{30}}{u_{31}} \\
		&&& b_1 b_2 \alpha^2 u_{30}u_{31} = (u_{30}-\lambda^{-1}w_2)(u_{31}-\mu^{-1}w_1)
	\end{aligned}
	\end{equation}
	Notice that the variable $u_{23}$ does not appear in any of the constraints, so it can essentially take any value. To solve this system we continue to eliminate variables. First, using the second constraint we eliminate $\alpha$
	\begin{equation}
		\alpha = -\frac{\lambda w_1}{\mu w_4}
	\end{equation}
	Then, using the fourth constraint we eliminate $\mu$
	\begin{equation}
		\mu = b_1 \lambda \frac{u_{31}}{u_{30}}
	\end{equation}
	Substituting these expressions into the first constraint gives
	\begin{equation}
		-w_2 = \left(-\frac{\lambda w_1}{\mu w_4}\right)\left(\frac{\lambda}{\mu}\right)w_3 = -\frac{w_1w_3}{w_4} \frac{u_{30}^2}{u_{31}^2}
	\end{equation}
	which allows us to eliminate $u_{31}$ with the introduction of a new variable $b_3\in\{-1,1\}$
	\begin{equation}
		u_{31} = b_3 \frac{w_1^\frac12 w_3^\frac12}{w_2^\frac12 w_4^\frac12} u_{30}
	\end{equation}
	We now note that
	\begin{equation}
		\alpha^2 u_{31}^2 = \frac{w_1^2}{w_4^2} \left(\frac{\lambda}{\mu}\right)^2 u_{31}^2 = \frac{w_1^2}{w_4^2} u_{30}^2
	\end{equation}
	and so the third constraint reads
	\begin{equation}\label{lambda_def}
	\begin{aligned}
		\frac{w_1^2}{w_4^2}u_{30}^2 &= \left( u_{30} - \lambda^{-1}w_2 \right)^2 \\
		\implies \lambda &= \frac{w_2}{u_{30}} \left( 1-b_4\frac{w_1}{w_4} \right)^{-1}
	\end{aligned}
	\end{equation}
	where $b_4$ is another new variable which takes the values $\pm1$. We are now left with only one complex variable, and one constraint. Before we deal with this constraint, note the following
	\begin{equation}
		\alpha^2 = \left(\frac\lambda\mu\right)^2\left(\frac{w_1^2}{w_4^2}\right) = \left(\frac{u_{30}^2}{u_{31}^2}\right) \left(\frac{w_1^2}{w_4^2}\right) = \left(\frac{w_2w_4}{w_1w_3}\right)\left(\frac{w_1^2}{w_4^2}\right) = \frac{w_1w_2}{w_3w_4}
	\end{equation}
	and
	\begin{equation}
		\mu u_{31} = b_1 \lambda \frac{u_{31}^2}{u_{30}} = b_1 w_2 \left(1-b_4\frac{w_1}{w_4}\right)^{-1} \frac{u_{31}^2}{u_{30}^2} = b_1 \frac{w_1w_3}{w_4} \left(1-b_4\frac{w_1}{w_4}\right)^{-1}
	\end{equation}
	The final constraint now reads
	\begin{equation}
	\begin{aligned}
		b_1b_2\left(\frac{w_1w_2}{w_3w_4}\right) u_{30}u_{31} &= \left(b_4 \frac{w_1}{w_4} u_{30} \right) \left(1-\frac{w_1}{\mu u_{31}}\right) u_{31} \\
		\implies b_1b_2\left(\frac{w_1w_2}{w_3w_4}\right) &= b_4\frac{w_1}{w_4} \left( 1-b_1\frac{w_4}{w_3}\left(1-b_4\frac{w_1}{w_4}\right)\right) \\
		\implies b_1b_2\frac{w_1w_2}{w_3w_4} &= b_4\frac{w_1}{w_4} - b_1b_4\frac{w_1}{w_3} + b_1\frac{w_1^2}{w_3w_4} \\
		\implies b_2w_2 &= b_1b_4 w_3 - b_4w_4 + w_1
	\end{aligned}
	\end{equation}
	Now we find that something remarkable has happened. Not only has the variable $u_{30}$ cancelled from this constraint, leaving it as another free variable, but also we are left with a constraint solely in terms of the $w_i$ and the signs $b_1,b_2,b_4$. When can this constraint be satisfied? Notice that the freedom we have in choosing $b_1,b_2,b_4$ allows us to choose, independently, whichever sign we wish ($\pm1$) in front of each of $w_2, w_3$ and $w_4$. Therefore, this constraint can be satisfied only when
	\begin{equation} \label{w_constraint}
		w_1 \pm w_2 \pm w_3 \pm w_4 = 0
	\end{equation}
	for some choice of signs.
	
	Moreover, if this constraint can be satisfied, then the original system of equations has a solution. To check this we need only substitute backwards our freely chosen values for $u_{23}, u_{30}$ and $b_3$. A problem can only occur where we encounter a division by zero (all the other operations we performed were reversible). Where could such a problem occur?
	\begin{itemize}
		\item If we tried setting $u_{23}=0$ or $u_{30}=0$ we would certainly encounter problems, since we have already remarked that solutions do not exist in this case.
		\item In defining $\lambda$ (equation \eqref{lambda_def}) we require that $1-b_4\frac{w_1}{w_4}$ is not zero. In fact this is never a problem. Looking at our original definitions of $w_1$ and $w_4$ we see that $\frac{w_1}{w_4} = -i\cot (\alpha-\beta)$ which is purely imaginary.
		\item $u_{23},u_{30}\neq0$ then implies that $\lambda,u_{31},\mu,\alpha,u_{22},u_{32},u_{33}$ are all trivially non-zero. In defining $u_{20}$ and $u_{21}$ (equations \eqref{u21_def} and \eqref{u20_def}) we require that they are not zero. This also is not a problem. Equation \eqref{u20u21notzero} has non-zero left-hand side, and hence neither of the terms on the right-hand side can be zero.
	\end{itemize}
	We conclude that any choice of $u_{23},u_{30}\neq0$ will lead to solutions of the original problem. Therefore, the unitary $W$ can be implemented if and only if the constraint \eqref{w_constraint} can be satisfied.
\end{proof}

\begin{lemma} \label{zero_case}
	A unitary $W$ of the form \eqref{cartan_form} with $w_i=0$ for some $i$ can always be achieved by the scheme in Figure \ref{circuit_1}.
\end{lemma}
\begin{proof}
	Assume that $w_1=0$ (the other cases are similar). We seek solutions to \eqref{simple_equation} which is a system of 16 polynomial equations in 16 variables. If we set
	\begin{equation}
		u_{00}=u_{02}=u_{20}=u_{22} = u_{11}=u_{13}=u_{31}=u_{33} = 0
	\end{equation}
	then many of our equations are trivially fulfilled. In fact, we are left with only 6 equations, in the remaining 8 variables
	\begin{equation}
	\begin{aligned}
		u_{10}u_{32} + u_{30}u_{12} &= w_4 \\
		u_{03}u_{30} &= w_2 \\
		u_{10}u_{23} &= w_3 \\
		u_{01}u_{32} &= w_3 \\
		u_{12}u_{21} &= w_2 \\
		u_{01}u_{23} + u_{03}u_{21} &= w_4
	\end{aligned}
	\end{equation}
	Further, setting
	\begin{equation}
		u_{32}=u_{30}=1
	\end{equation}
	implies
	\begin{equation}
		u_{01}=w_3, u_{03}=w_2
	\end{equation}
	and reduces the problem to 4 equations in 4 unknowns
	\begin{equation} \label{four_equations}
	\begin{aligned}
		u_{12}u_{21} &= w_2 \\
		u_{10}u_{23} &= w_3 \\
		u_{10} + u_{12} &= w_4 \\
		w_3u_{23} + w_2u_{21} &= w_4
	\end{aligned}
	\end{equation}
	Suppose that $w_2,w_3\neq0$ (and notice that $w_1=0$ implies $\cos(\alpha-\beta)=0$ which implies $w_4\neq0$). Set
	\begin{equation} \label{def_u23}
		u_{23}=\frac{w_3}{u_{10}}, u_{21}=\frac{w_2}{u_{12}}, u_{12}=w_4-u_{10}
	\end{equation}
	Substituting these into the final equation of \eqref{four_equations} gives
	\begin{equation}
		\frac{w_3^2}{u_{10}} + \frac{w_2^2}{w_4-u_{10}} = w_4
	\end{equation}
	which rearranges to
	\begin{equation} \label{u10_quadratic}
		w_4 u_{10}^2 + \left(w_2^2-w_3^2-w_4^2\right) u_{10} + w_3^2 w_4 = 0
	\end{equation}
	This is a quadratic equation for $u_{10}$ which must have at least one complex root. Furthermore, we know that $0$ and $w_4$ are not roots of \eqref{u10_quadratic} because setting $u_{10}=0$ and $u_{10}=w_4$ in \eqref{u10_quadratic} gives $w_3^2 w_4=0$ and $w_2^2 w_4=0$, both of which do not hold. Therefore, if we choose any root of \eqref{u10_quadratic} for $u_{10}$ and substitute this value back into \eqref{def_u23} we obtain a solution to \eqref{simple_equation}.
	
	The final case we need to consider is when one of $w_2,w_3$ is zero. Assume that $w_2=0$ (the case $w_3=0$ is similar). This implies that $\cos(\alpha+\beta)=0$ and hence that $w_3\neq0$. Then
	\begin{equation}
		u_{21}=0, u_{23}=\frac{w_4}{w_3}, u_{10}=\frac{w_3^2}{w_4}, u_{12}=w_4-\frac{w_3^2}{w_4}
	\end{equation}
	is a solution to \eqref{four_equations} and we have a solution to \eqref{simple_equation}.
\end{proof}

\bibliography{References}

\begin{thebibliography}{10}

\bibitem{KLM01}
E.~Knill, R.~Laflamme, and G.J. Milburn.
\newblock A scheme for efficient quantum computation with linear optics.
\newblock {\em Nature}, 409:46--52, 2001.

\bibitem{RLB02}
T.C. Ralph, N.K. Langford, T.B. Bell, and A.G. White.
\newblock Linear optical controlled-{NOT} gate in the coincidence basis.
\newblock {\em Phys. Rev. A}, 65(6):062324, 2002.

\bibitem{HT02}
H.F. Hofmann and S.~Takeuchi.
\newblock Quantum phase gate for photonic qubits using only beam splitters and
  postselection.
\newblock {\em Phys. Rev. A}, 66(2):024308, 2002.

\bibitem{OPW03}
J.L. O'Brien, G.J. Pryde, A.G. White, T.C. Ralph, and D.~Branning.
\newblock Demonstration of an all-optical quantum controlled-{NOT} gate.
\newblock {\em Nature}, 426:264--267, 2003.

\bibitem{SVP12}
P.J. Shadbolt, M.R. Verde, A.~Peruzzo, A.~Politi, A.~Laing, M.~Lobino, J.C.F.
  Matthews, M.G. Thompson, and J.L. O'Brien.
\newblock Generating, manipulating and measuring entanglement and mixture with
  a reconfigurable photonic circuit.
\newblock {\em Nature Photon.}, 6:45--49, 2012.

\bibitem{LPN11}
H.W. Li, S.~Przeslak, A.O. Niskanen, J.C.F. Matthews, A.~Politi, P.~Shadbolt,
  A.~Laing, M.~Lobino, M.G. Thompson, and J.L. O'Brien.
\newblock Reconfigurable controlled two-qubit operation on a quantum photonic
  chip.
\newblock {\em New J. Phys.}, 13:115009, 2011.

\bibitem{PMO09}
A.~Politi, J.C.F. Matthews, and J.L. O'Brien.
\newblock Shor's quantum factoring algorithm on a photonic chip.
\newblock {\em Science}, 325:1221, 2009.

\bibitem{MLL12}
E.~Mart\'{i}n-L\'{o}pez, A.~Laing, T.~Lawson, R.~Alvarez, X.-Q. Zhou, and J.L.
  O'Brien.
\newblock Experimental realization of shor's quantum factoring algorithm using
  qubit recycling.
\newblock {\em Nature Photon.}, 6:773--776, 2012.

\bibitem{KOE10}
K.~Kieling, J.L. O'Brien, and J.~Eisert.
\newblock On photonic controlled phase gates.
\newblock {\em New J. Phys.}, 12:013003, 2010.

\bibitem{Kie08}
K.~Kieling.
\newblock {\em Linear optics quantum computing - construction of small networks
  and asymptotic scaling}.
\newblock PhD thesis, Imperial College London, 2008.

\bibitem{RZBB94}
M.~Reck, A.~Zeilinger, H.J. Bernstein, and P.~Bertani.
\newblock Experimental {R}ealization of {A}ny {D}iscrete {U}nitary {O}perator.
\newblock {\em Phys. Rev. Lett.}, 73(1):58--61, 1994.

\bibitem{KLD02}
P.~Kok, H.~Lee, and J.P. Dowling.
\newblock Single-photon quantum-nondemolition detectors constructed with linear
  optics and projective measurements.
\newblock {\em Phys. Rev. A}, 66(6):063814, 2002.

\bibitem{KC01}
B.~Kraus and J.I.Cirac.
\newblock Optimal creation of entanglement using a two-qubit gate.
\newblock {\em Phys. Rev. A}, 63(6):062309, 2001.

\bibitem{BdV08}
M.~Blaauboer and R.L. de~Visser.
\newblock An analytical decomposition protocol for optimal implementation of
  two-qubit entangling gates.
\newblock {\em J. Phys. A: Math. Theor.}, 41(39):395307, 2008.

\bibitem{NW99}
J.~Nocedal and S.J. Wright.
\newblock {\em Numerical {O}ptimization}.
\newblock Springer Verlag, 1999.

\bibitem{SciPy}
Eric Jones, Travis Oliphant, Pearu Peterson, et~al.
\newblock {SciPy}: Open source scientific tools for {Python}, 2001--.

\end{thebibliography}
\bibliographystyle{unsrt}

\end{document}